\documentclass[twoside,11pt,tablecaption=bottom]{jmlr}

\usepackage{times}

\usepackage{mathtools}
\usepackage{enumitem}
\newlist{inlistalph}{enumerate*}{1}
\setlist[inlistalph]{label=(\alph*)}

\usepackage[utf8]{inputenc}
\usepackage[english]{babel}
\usepackage{amsopn}
\usepackage{amsfonts}
\usepackage{mathrsfs}
\usepackage{xparse}
\usepackage{upgreek}

\usepackage{wasysym}
\usepackage{color}
\usepackage{tabu}

\usepackage{float}

\usepackage{verbatim}

\usepackage{adjustbox}

\newcommand*{\qedhere}{\hfill\BlackBox\\[2mm]}
\newcommand*{\claimqedhere}{\hfill (Claim) \BlackBox\\[2mm]}

\newcommand{\AND}{{\mbox{ }\wedge\mbox{ }}}

\newcommand{\NAT}{\mathbb{N}}

\newcommand*{\concat}{^\frown}

\usepackage{tikz}
\usepackage{xspace}
\usetikzlibrary{shapes, fit, decorations.pathreplacing, shadows, arrows, calc, positioning}

\usepackage{pgfplots}
\usepgfplotslibrary{patchplots}

\newcommand*{\Txt}{\mathbf{Txt}}

\newcommand*{\G}{\mathbf{G}}

\newcommand*{\It}{\mathbf{It}}
\newcommand*{\Sd}{\mathbf{Sd}}
\newcommand*{\Psd}{\mathbf{Psd}}
\newcommand*{\Td}{\mathbf{Td}}

\newcommand*{\Ex}{\mathbf{Ex}}
\newcommand*{\Bc}{\mathbf{Bc}}

\newcommand*{\Fin}{\mathbf{Fin}}

\newcommand*{\True}{\mathbf{T}}

\newcommand*{\Cind}{\mathbf{CInd}}
\newcommand*{\T}{\mathbf{T}}

\newcommand*{\N}{\mathbb{N}}

\newcommand*{\Pow}{\mathrm{Pow}}

\newcommand*{\La}{\mathcal{L}}

\newcommand*{\Seq}{{\mathbb{S}\mathrm{eq}}}

\newcommand*{\totalCp}{\mathcal{R}}
\newcommand*{\partialCp}{\mathcal{P}}

\newcommand{\range}{\mathrm{range}}
\newcommand{\content}{\mathrm{content}}

\newcommand{\ind}{\mathrm{ind}}
\newcommand{\pad}{\mathrm{pad}}

\newcommand{\ORT}{\textbf{ORT}\xspace}

\newcommand{\settwo}[1]{
  \{ #1 \}
}

\makeatletter
\newsavebox{\@brx}
\newcommand{\llangle}[1][]{\savebox{\@brx}{\(\m@th{#1\langle}\)}%
  \mathopen{\copy\@brx\kern-0.5\wd\@brx\usebox{\@brx}}}
\newcommand{\rrangle}[1][]{\savebox{\@brx}{\(\m@th{#1\rangle}\)}%
  \mathclose{\copy\@brx\kern-0.5\wd\@brx\usebox{\@brx}}}
\makeatother

\newcommand{\convs}{\mathclose{\hbox{$\downarrow$}}}
\newcommand{\divs}{\mathclose{\hbox{$\uparrow$}}}

\newcommand*{\falls}{\text{if }}
\newcommand*{\sonst}{\text{otherwise}}
\newcommand*{\sonstfalls}{\text{else, if }}
\newcommand*{\cIf}{\text{if }}
\newcommand*{\otw}{\text{otherwise}}
\newcommand{\abs}[1]{\lvert #1 \rvert}

\newcommand*{\noqed}{\renewcommand{\jmlrQED}{}}

\newcommand{\itemin}[1]{\item[#1\hspace{-0.5cm}] \hspace{0.5cm}}

\usepackage{thm-restate}
\newcounter{count}
\newtheorem{claim}[count]{Claim}

\makeatletter 
\g@addto@macro{\@algocf@init}{\SetKwInput{input}{Input}} 
\makeatother
\makeatletter 
\g@addto@macro{\@algocf@init}{\SetKwInput{output}{Semantic Output}} 
\makeatother
\makeatletter 
\g@addto@macro{\@algocf@init}{\SetKwInput{init}{Initialization}} 
\makeatother

\newcommand*{\rec}{_\mathbf{REC}}
\newcommand{\tORT}{Operator Recursion Theorem (\(\ORT\))}

\newcommand*{\CalL}{\mathcal{L}}
\newcommand*{\CalR}{\mathcal{R}}

\newcommand*{\natnum}{\mathbb{N}}

\usepackage{xparse}
\DeclareDocumentCommand{\set}{m g o}%
{%
    \IfNoValueTF{#3}{\left}{#3}\{#1
            \IfNoValueTF{#2}{}{\ \IfNoValueTF{#3}{\left}{#3}\vert\ \vphantom{#1}#2\IfNoValueTF{#3}{\right.}{}}
                \IfNoValueTF{#3}{\right}{#3}\}%
}

\DeclareDocumentCommand{\abs}{m o}%
{%
    \IfNoValueTF{#2}{\left}{#2}\vert#1
                \IfNoValueTF{#2}{\right}{#2}\vert%
}
\usepackage{hyperref}
\usepackage[capitalise, noabbrev]{cleveref}
\pgfplotsset{compat=1.16}

\newcommand*{\REC}{\textbf{REC}}

\newcommand{\sort}{\operatorname{sort}}

\allowdisplaybreaks

\jmlrvolume{}
\jmlryear{}
\jmlrsubmitted{}
\jmlrpublished{}
\jmlrworkshop{}
\jmlrpages{}
\jmlrproceedings{}{}

\title[Learning Languages with Decidable Hypotheses]{Learning Languages with Decidable Hypotheses}

\author{
\Name{Julian Berger} \Email{julian.berger@student.hpi.uni-potsdam.de} \\ 
\Name{Maximilian B\"{o}ther} \Email{maximilian.boether@student.hpi.uni-potsdam.de} \\ 
\Name{Vanja Dosko\v{c}} \Email{vanja.doskoc@hpi.de} \\ 
\Name{Jonathan Gadea Harder} \Email{jonathan.gadeaharder@student.hpi.uni-potsdam.de} \\ 
\Name{Nicolas Klodt} \Email{nicolas.klodt@student.hpi.uni-potsdam.de} \\ 
\Name{Timo K\"{o}tzing} \Email{timo.koetzing@hpi.de} \\  
\Name{Winfried L\"{o}tzsch} \Email{winfried.loetzsch@student.hpi.uni-potsdam.de} \\ 
\Name{Jannik Peters} \Email{jannik.peters@student.hpi.uni-potsdam.de} \\ 
\Name{Leon Schiller} \Email{leon.schiller@student.hpi.uni-potsdam.de} \\ 
\Name{Lars Seifert} \Email{lars.seifert@student.hpi.uni-potsdam.de} \\ 
\Name{Armin Wells} \Email{armin.wells@student.hpi.uni-potsdam.de} \\ 
\Name{Simon Wietheger} \Email{simon.wietheger@student.hpi.uni-potsdam.de} \\ 
\addr Hasso Plattner Institute \\ University of Potsdam, Germany}

\begin{document}

\maketitle

\begin{abstract}
	In \emph{language learning in the limit}, the most common type of hypothesis
    is to give an enumerator for a language. This so-called $W$-index allows for
    naming arbitrary computably enumerable languages, with the drawback that
    even the membership problem is undecidable. In this paper we use a different
    system which allows for naming arbitrary decidable languages, namely
    \emph{programs for characteristic functions} (called $C$-indices). These
    indices have the drawback that it is now not decidable whether a given
    hypothesis is even a legal $C$-index.
	
	In this first analysis of learning with $C$-indices, we give a structured
    account of the learning power of various restrictions employing $C$-indices,
    also when compared with $W$-indices. We establish a hierarchy of learning
    power depending on whether $C$-indices are required (a) on all outputs; (b)
    only on outputs relevant for the class to be learned and (c) only in the
    limit as final, correct hypotheses. Furthermore, all these settings are
    weaker than learning with $W$-indices (even when restricted to classes of
    computable languages). We analyze all these questions also in relation to
    the mode of data presentation.
	
	Finally, we also ask about the relation of semantic versus syntactic
    convergence and derive the map of pairwise relations for these two kinds of
    convergence coupled with various forms of data presentation.
\end{abstract}
\begin{keywords}
  language learning in the limit, inductive inference, decidable languages, characteristic index
\end{keywords}

\section{Introduction}

We are interested in the problem of algorithmically learning a description for a
formal language (a computably enumerable subset of the set of natural numbers)
when presented successively all and only the elements of that language; this is
called \emph{inductive inference}, a branch of (algorithmic) learning theory.
For example, a learner $h$ might be presented more and more even numbers. After
each new number, $h$ outputs a description for a language as its conjecture. The
learner $h$ might decide to output a program for the set of all multiples of
$4$, as long as all numbers presented are divisible by~$4$. Later, when $h$ sees
an even number not divisible by $4$, it might change this guess to a program for
the set of all multiples of~$2$.

Many criteria for determining whether a learner $h$ is \emph{successful} on a
language~$L$ have been proposed in the literature. \citet{Gold67}, in his
seminal paper, gave a first, simple learning criterion,
\emph{$\Txt\G\Ex$-learning}\footnote{$\Txt$ stands for learning from a
\emph{text} of positive examples; $\G$ for Gold, indicating full-information
learning; $\Ex$ stands for \emph{explanatory}.}, where a learner is
\emph{successful} if and only if, on every \emph{text} for $L$ (listing of all
and only the elements of $L$) it eventually stops changing its conjectures, and
its final conjecture is a correct description for the input language.

Trivially, each single, describable language $L$ has a suitable constant
function as a $\Txt\G\Ex$-learner (this learner constantly outputs a description
for $L$). Thus, we are interested in analyzing for which \emph{classes of
languages} $\CalL$ is there a \emph{single learner} $h$ learning \emph{each}
member of $\CalL$. This framework is also known as \emph{language learning in
the limit} and has been studied extensively, using a wide range of learning
criteria similar to $\Txt\G\Ex$-learning (see, for example, the textbook
\citet{JORS99}).

In this paper we  put the focus on the possible descriptions for languages. Any
computably enumerable language $L$ has as possible descriptions any program
enumerating all and only the elements of $L$, called a $W$-index (the language
enumerated by program $e$ is denoted by $W_e$). This system has various
drawbacks; most importantly, the function which decides, given $e$ and $x$,
whether $x \in W_e$ is not computable. We propose to use different descriptors
for languages: programs for characteristic functions (where such programs $e$
describe the language $C_e$ which it decides). Of course, only decidable
languages have such a description, but now, given a program $e$ for a
characteristic function, $x \in C_e$ is decidable. Additionally to many
questions that remain undecidable (for example whether $C$-indices are for the
same language or whether a $C$-index is for a finite language), it is not
decidable whether a program $e$ is indeed a program for a characteristic
function. Thich leads to a new set of problems: learners cannot be
(algorithmically) checked whether their outputs are viable (in the sense of
being programs for characteristic functions).

Based on this last observation we study a range of different criteria which
formalize what kind of behavior we expect form our learners. In the most relaxed
setting, learners may output any number (for a program) they want, but in order
to $\Ex$-learn, they need to to converge to a correct $C$-index; we denote this
restriction with $\Ex_C$. Requiring additionally to only use $C$-indices in
order to successfully learn we denote by $\Cind\Ex_C$; requiring $C$-indices on
\emph{all} inputs (not just for successful learning, but also when seeing input
from no target language whatsoever) we denote by $\tau(\Cind)\Ex_C$. In
particular, the last restriction requires the learner to be total; in order to
distinguish whether the loss of learning power is due to the totality
restriction or truly due to the additional requirement of outputting
$C$-indices, we also study $\CalR\Cind\Ex_C$, that is, the requirement
$\Cind\Ex_C$ where additionally the learner is required to be total.

We note that $\tau(\Cind)\Ex_C$ is similar to learning \emph{indexable
families}. Indexable families are classes of languages $\CalL$ such that there
is an enumeration $(L_i)_{i \in \natnum}$ of all and only the elements of
$\CalL$ for which the decision problem ``$x \in L_i$'' is decidable. Already for
such classes of languages we get a rich structure. A survey of previous work in
this area can be found in \citet{LZZ08}. For a learner $h$ learning according to
$\tau(\Cind)\Ex_C$ we have that $L_x = C_{h(x)}$ gives an indexing of a family
of languages, and $h$ learns some subset thereof. We are specifically interested
in the area between this setting and learning with $W$-indices ($\Ex_W$). 

The criteria we analyze naturally interpolate between these two settings. We
show that we have the following hierarchy: $\tau(\Cind)\Ex_C$ allows for
learning strictly fewer classes of languages than $\CalR\Cind\Ex_C$, which allow
for learning the \emph{same} classes as $\Cind\Ex_C$, which again are fewer than
learnable by $\Ex_C$, which in turn renders fewer classes learnable than
$\Ex_W$.

All these results hold for learning with full information. In order to study the
dependence on the mode of information presentation, we also consider
\emph{partially set-driven} learners ($\Psd$,~\citet{BlumBlum75,SchRicht84}),
which only get the set of data presented so far and the iteration number as
input; \emph{set-driven} learners ($\Sd$,~\citet{WC80}), which get only the set
of data presented so far; \emph{iterative} learners
($\It$,~\citet{Wiehagen76,Fulk85}), which only get the new datum and its current
hypothesis and, finally, \emph{transductive} learners
($\Td$,~\citet{CCJS07,Kotzing09}), which only get the current data. Note that
transductive learners are mostly of interest as a proper restriction to all
other modes of information presentation.

We show that full-information learners can be turned into partially set-driven
learners without loss of learning power. Furthermore, iterative learning is
strictly less powerful than set-driven learning, in all settings. Altogether we
analyze 25 different criteria and show how each pair relates. All these results
are summarized in Figure~\ref{Map:DifferentSettings} as one big map stating all
pairwise relations of the learning criteria mentioned, giving 300 pairwise
relations in one diagram, proven with 13 theorems in
Section~\ref{Sec:CindLearning}. Note that the results comparing learning
criteria with $W$-indices were previously known, and some proofs could be
extended to also cover learning with $C$-indices.

\begin{figure}[h]
\begin{center}
    \begin{adjustbox}{width=0.85\textwidth}
        \includegraphics{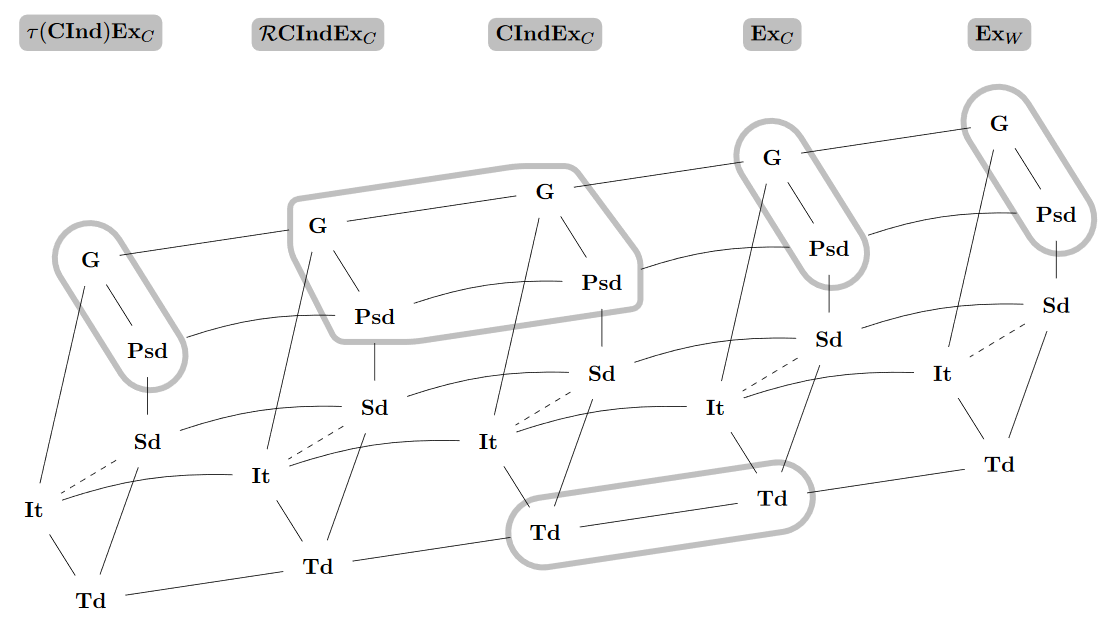}
    \end{adjustbox}
\end{center}
    \caption{Relation of various requirements when to output characteristic
        indices paired with various memory restrictions. We omit mentioning $\Txt$
        to favour readability. Black solid lines imply trivial inclusions
        (bottom-to-top, left-to-right). Dashed lines imply non-trivial inclusions
        (bottom-to-top, left-to-right). Furthermore, greyly edged areas illustrate a
        collapse of the enclosed learning criteria and there are no further collapses.}
    \label{Map:DifferentSettings}
\end{figure}

We derive a similar map considering a possible relaxation on $\Ex_C$-learning:
while $\Ex_C$ requires syntactic convergence to one single correct $C$-index, we
consider \emph{behaviorally correct} learning, $\Bc_C$ for short, where the
learner only has to semantically converge to correct $C$-indices (but may use
infinitely many different such indices). We again consider the different modes
of data presentation and determine all pairwise relations in
Figure~\ref{Map:Convergence}.

\begin{figure}[h]
\begin{center}
    \begin{adjustbox}{width=0.42\textwidth}
        \includegraphics{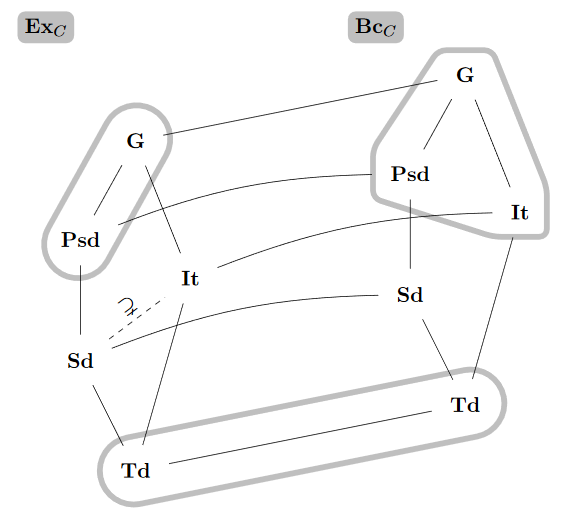}
    \end{adjustbox}
\end{center}
    \caption{Relation of learning criteria under various memory restrictions. On
        the left-hand side we require syntactic convergence to characteristic
        indices and semantic on the right-hand side. We omit mentioning $\Txt$ in favour
        of readability. Black solid lines imply trivial inclusions (bottom-to-top,
        left-to-right). The dashed line between $\Sd$ and $\It$ indicates that
        $[\Txt\It\Ex_C]\rec \subsetneq [\Txt\Sd\Ex_C]\rec$. Furthermore, greyly edged
        areas illustrate a collapse of the enclosed learning criteria and there are no
        further collapses.}
    \label{Map:Convergence}
\end{figure}

Before getting to our results in detail, we continue with some (mathematical)
preliminaries in Section~\ref{sec:preliminaries}.

\section{Preliminaries}%
\label{sec:preliminaries}
In this section we discuss the used notation. Unintroduced notation follows the
textbook of~\citet{Rogers87}. For learning criteria we follow the system
introduced by \citet{Kotzing09}. 

\subsection{Mathematical Notations and Learning Criteria}%
\label{sub:mathematical_notatation_and_language_learning}

With \(\N\) we denote the set of all natural numbers, namely
\(\set{0,1,2,\ldots} \). We denote the subset and proper subset relation between
two sets with \(\subseteq\) and \(\subsetneq\), respectively. With
\(\subseteq_\Fin\) we denote the finite subset relation. We use \(\emptyset\)
and \(\varepsilon\) to denote the empty set and empty sequence, respectively.
For any set \(A\), the set of all subsets of \(A\) is denoted by \(\Pow(A)\).
The set of all computable functions is denoted by \(\partialCp\), the subset of
all total computable functions by \(\totalCp\). If a function \(f\) is (not)
defined on some argument \(x\in\N\), we say that \(f\) converges (diverges) on
\(x\), denoting this fact with \(f(x)\convs\) (\(f(x)\divs\)). We fix an
effective numbering \(\set{{\varphi_{e}}}_{e\in\N} \) of \(\partialCp\). For any
\(e\in\N\), we let \(W_e\) denote the domain of \(\varphi_{e}\) and call \(e\) a
\(W\)-index of \(W_e\). This set we call the \emph{$e$-th computably enumerable
set}. We call \(e\in\N\) a \emph{\(C\)-index} (\emph{characteristic index}) if
and only if \(\varphi_{e}\) is a total function such that for all \(x\in\N\) we
have \(\varphi_{e}(x)\in \set{0,1} \). Furthermore, we let
\(C_e=\set{x\in\N}{\varphi_{e}(x)=1 }\). For a computably enumerable set $L$, if
some \(e\in\N\) is a \(C\)-Index with \(C_e=L\), we write
\(\varphi_{e}=\chi_{L}\). Note that, if a set has a \(C\)-index, it is
\emph{recursive}.  The set of all recursive sets is denoted by \(\REC\). For a
finite set \(D\subseteq\N\), we let \(\ind(D)\) be a \(C\)-index for \(D\). Note
that \(\ind\in\totalCp\). Furthermore, we fix a Blum complexity measure \(\Phi\)
associated with \(\varphi\), that is, for all \(e,x\in\N\), \(\Phi_e(x)\) is the
number of steps the function $\varphi_e$ takes on input \(x\) to converge,
see~\citep{Blum67}. The padding function \(\pad\in\totalCp\) is an injective
function such that, for all \(e,n\in\N\), we have \(
\varphi_{e}=\varphi_{_{\pad(e,n)}}\). We use \(\langle \cdot,\cdot \rangle \) as
a computable, bijective function that codes a pair of natural numbers into a
single one. We use \(\pi_1\) and \(\pi_2\) as computable decoding functions for
the first and section component, i.e., for all \(x,y\in\N\) we have \(
\pi_1(\langle x,y \rangle )=x\) and \(\pi_2(\langle x,y\rangle )=y\). 

We learn computably enumerable sets $L$, called \emph{languages}. We fix a
\emph{pause} symbol \(\#\), and let, for any set \(S\), \(S_\# \coloneqq S\cup
\set{\#} \). Information about languages is given from \emph{text}, that is,
total functions \(T\colon \N\to\N\cup \set{\#} \). A text $T$ is of a certain
language $L$ if its \emph{content} is exactly $L$, that is, $\content(T)
\coloneqq \range(T) \setminus \{ \# \}$ is exactly $L$. We denote the set of all
texts as $\Txt$ and the set of all texts of a language $L$ as $\Txt(L)$. For any
$n \in \N$, we denote with $T[n]$ the initial sequence of the text $T$ of length
$n$, that is, \(T[0] \coloneqq \varepsilon\) and $T[n] \coloneqq (T(0), \dots,
T(n-1))$. Given a language \(L\) and  \(t\in\N\), the set of sequences
consisting of elements of \(L \cup \{ \# \}\) that are at most \(t\) long, is
denoted by \(L^{\leq t}_\#\). Furthermore, we denote with \(\Seq\) all finite
sequences over \(\N_\# \) and define the \emph{content} of such sequences
analogous to the content of texts. The concatenation of two sequences \(\sigma,
\tau\in\Seq\) is denoted by  \(\sigma\concat\tau\). Furthermore, we write
$\subseteq$ for the \emph{extension relation} on sequences and fix a order
$\leq$ on \(\Seq\) interpreted as natural numbers.

Now, we formalize learning criteria following the system introduced by
\citet{Kotzing09}. A \emph{learner} is a partial function \(h\in\partialCp\). An
\emph{interaction operator} \(\beta\) is an operator that takes a learner
\(h\in\partialCp\) and a text \(T\in\Txt\) as input and outputs a (possibly
partial) function \(p\). Intuitively, \(\beta\) defines which information is
available to the learner for making its hypothesis. We consider
\emph{Gold-style} or \emph{full-information} learning (\citet{Gold67}), denoted
by \(\G\), \textit{partially set-driven} learning
(\(\Psd\),~\citet{BlumBlum75,SchRicht84}), \textit{set-driven} learning
(\(\Sd\),~\citet{WC80}), \textit{iterative} learning (\(\It\),~\citet{
Wiehagen76,Fulk85}) and \textit{transductive} learning
(\(\Td\),~\citet{CCJS07,Kotzing09}). To define the latter formally, we introduce
a symbol ``?'' for the learner to signalize, that the information given is
insufficient. Formally, for all learners \(h\in\partialCp\), texts \(T\in\Txt\)
and all \(i\in\N\), define
\begin{align*}
    \G(h,T)(i) &= h(T[i]);\\
    \Psd(h,T)(i) &= h(\content(T[i]),i);\\
    \Sd(h,T)(i) &= h(\content(T[i]));\\
    \It(h,T)(i) &= 
        \begin{cases}
            h(\varepsilon), &\cIf i=0;\\
            h(\It(h,T)(i-1),T(i-1)), &\otw;
        \end{cases}\\
    \Td(h,T)(i) &=
        \begin{cases}
            \mbox{?}, &\cIf i=0;\\
            \Td(h,T)(i-1), &\sonstfalls h(T(i-1))=\mbox{?};\\
            h(T(i-1)), &\otw.
        \end{cases}
\end{align*}
For any of the named interaction operators \(\beta\), given a \(\beta\)-learner
\(h\), we let \(h^*\) (the \textit{starred} learner) denote a \(\G\)-learner
simulating \(h\), i.e., for all \(T\in\Txt\), we have \(\beta(h,T)=\G(h^*,T)\).
For example, let \(h\) be a \(\Sd\)-learner. Then, intuitively, \(h^*\) ignores
all information but the content of the input, simulating \(h\) with this
information, i.e., for all finite sequences \(\sigma\), we have
\(h^*(\sigma)=h(\content(\sigma))\).

For a learner to successfully identify a language, we may oppose constraints on
the hypotheses the learner makes.  These are called \emph{learning
restrictions}. A famous example was given by~\citet{Gold67}. He required the
learner to be \emph{explanatory}, i.e., the learner must converge to a
\emph{single}, correct hypothesis  for the target language. We hereby
distinguish whether the final hypothesis is interpreted as a \(C\)-index or as a
\(W\)-index, denoting this by \(\Ex_C\) and \(\Ex_W\), respectively. Formally,
for any sequence of hypotheses \(p\) and text \(T\in\Txt\), we have
\begin{align*}
     \Ex_C(h,T)&\Leftrightarrow\exists n_0\colon\forall n\geq n_0\colon p(n)=p(n_0)\land
        \varphi_{p(n_0)}=\chi_{\content(T)};\\
     \Ex_W(h,T)&\Leftrightarrow\exists n_0\colon\forall n\geq n_0\colon p(n)=p(n_0)\land
        W_{p(n_0)}=\content(T).
\end{align*}
We say that explanatory learning requires \emph{syntactic} convergence. If there
exists a \(C\)-index (or \(W\)-index) for a language, then there exist
infinitely many. This motivates to not require syntactic but only
\emph{semantic} convergence, i.e., the learner may make mind changes, but it
has to, eventually, only output correct hypotheses. This is called
\emph{behaviorally correct} learning (\(\Bc_C\) or \(\Bc_W\),~\citet{CL82,OW82}).
Formally, let \(p\) be a sequence of hypotheses and let \(T\in\Txt\), then
\begin{align*}
     \Bc_C(p, T)&\Leftrightarrow\exists n_0\colon\forall n\geq
        n_0\colon\varphi_{p(n)}=\chi_{\content(T)};\\
     \Bc_W(p, T)&\Leftrightarrow\exists n_0\colon\forall n\geq
        n_0\colon W_{p(n)}=\content(T).
\end{align*}
In this paper, we consider learning with \(C\)-indices. It is, thus, natural to
require the hypotheses to consist solely of \(C\)-indices, called
\emph{\(C\)-index} learning, and denoted by \(\Cind\). Formally, for a sequence
of hypotheses \(p\) and a text \(T\), we have
\begin{align*}
     \Cind(p,T)&\Leftrightarrow\forall i,x\colon
        \varphi_{p(i)}(x)\in\set{0,1}.
\end{align*}
For two learning restrictions \(\delta\) and \(\delta'\), their combination is
their intersection, denoted by their juxtaposition \(\delta\delta'\). We let
\(\True\) denote the learning restriction that is always true, which is
interpreted as the absence of a learning restriction.

A \textit{learning criterion} is a tuple \((\alpha,\mathcal{C},\beta,\delta)\),
where \(\mathcal{C}\) is the set of admissible learners, usually \(\partialCp\)
or \(\totalCp\), \(\beta\) is an interaction operator and \(\alpha\) and
\(\delta\) are learning restrictions. We denote this criterion with
\(\tau(\alpha)\mathcal{C}\Txt\beta\delta\), omitting \(\mathcal{C}\) if
\(\mathcal{C}=\partialCp\), and a learning restriction if it equals \(\True\).
We say that an admissible learner \(h\in\mathcal{C}\)
\(\tau(\alpha)\mathcal{C}\Txt\beta\delta\)-learns a language \(L\) if and only
if, for arbitrary texts \(T\in\Txt\), we have \(\alpha(\beta(h,T),T)\) and for
all texts \(T\in\Txt(L)\) we have \(\delta(\beta(h,T),T)\). The set of languages
\(\tau(\alpha)\mathcal{C}\Txt\beta\delta\)-learned by  \(h\in\mathcal{C}\) is
denoted by \(\tau(\alpha)\mathcal{C}\Txt\beta\delta(h)\). With
\([\tau(\alpha)\mathcal{C}\Txt\beta\delta]\) we denote the set of all classes
\(\tau(\alpha)\mathcal{C}\Txt\beta\delta\)-learnable by some learner in
\(\mathcal{C}\). Moreover, to compare learning with $W$- and $C$-indices, these
classes may only contain recursive languages, which we denote as
\({[\tau(\alpha)\mathcal{C}\Txt\beta\delta]}\rec\). 

\subsection{Normal Forms}%
\label{sub:normal_forms}
When studying language learning in the limit, there are certain properties of
learner that are useful, e.g., if we can assume a learner to be total.
\citet{KP16} and \citet{KSS17} study under which circumstances learners may be
assumed to be total. Importantly, this is the case for explanatory Gold-style
learners obeying delayable learning restrictions and for behaviorally correct
learners obeying delayable restrictions. Intuitively, a learning restriction is
\emph{delayable} if it allows hypotheses to be arbitrarily, but not indefinitely
postponed without violating the restriction.  Formally, a learning restriction
\(\delta\) is delayable, if and only if for all non-decreasing, unbounded
functions \(r\colon\N\to\N\), texts \(T, T'\in\Txt\) and learning sequences
\(p\) such that for all \(n\in\N\), \(\content(T[r(n)])\subseteq
\content(T'[n])\) and \(\content(T)=\content(T')\), we have, if \(\delta(p,T)\),
then also  \(\delta(p\circ r, T')\). Note that $\Ex_W$, $\Ex_C$, $\Bc_W$,
$\Bc_C$ and $\Cind$ are delayable restrictions.

Another useful notion are \emph{locking sequences}. Intuitively, these contain
enough information such that a learner, after seeing this information, converges
correctly and does not change its mind anymore whatever additional information
from the target language it is given. Formally, let \(L\) be a language and let
\(\sigma\in L_\#^*\). Given a \(\G\)-learner \(h\in\partialCp\), \(\sigma\) is a
\emph{locking sequence} for \(h\) on \(L\) if and only if for all sequences
\(\tau\in L_\#^*\) we have \(h(\sigma)=h(\sigma\tau)\) and \(h(\sigma)\) is a
correct hypothesis for \(L\), see~\citet{BlumBlum75}.  This concept can
immediately be transferred to other interaction operators. Exemplary, given a
$\Sd$-learner $h$ and a locking sequence $\sigma$ of the starred learner $h^*$,
we call the set $\content(\sigma)$ a \emph{locking set}. Analogously, one
transfers this definition to the other interaction operators. It shall not
remain unmentioned that, when considering $\Psd$-learners, we speak of
\emph{locking information}.  In the case of \(\Bc_W\)-learning we do not require
the learner to syntactically converge. Therefore, we call a sequence \(\sigma\in
L_\#^*\) a \emph{\(\Bc_W\)-locking sequence} for a  \(\G\)-learner \(h\) on
\(L\) if, for all sequences \(\tau \in L_\#^*\), \(h(\sigma\tau)\) is a correct
hypothesis for  \(L\), see~\citet{JORS99}. We omit the transfer to other
interaction operators as it is immediate. It is an important observation by
\citet{BlumBlum75}, that for any learner $h$ and any language $L$ it learns,
there exists a ($\Bc_W$-) locking sequence. These notions and results directly
transfer to $\Ex_C$- and $\Bc_C$-learning. When it is clear from the context, we
omit the index.

\section{Requiring $C$-Indices as Output}\label{Sec:CindLearning}

This section is dedicated to proving Figure~\ref{Map:DifferentSettings}, giving
all pairwise relations for the different settings of requiring $C$-indices for
output in the various mentioned modes of data presentation. In general, we
observe that the later we require $C$-indices, the more learning power the
learner has. This holds except for transductive learners which converge to
$C$-indices. We show that they are as powerful as $\Cind$-transductive learners. 

Although we learn classes of recursive languages, the requirement to converge to
characteristic indices does heavily limit a learners capabilities. In the next
theorem we show that even transductive learners which converge to $W$-indices
can learn classes of languages which no Gold-style $\Ex_C$-learner can learn. We
exploit the fact that $C$-indices, even if only conjectured eventually, must
contain both positive and negative information about the guess. 

\begin{theorem}
  We have that \({[\Txt\Td\Ex_W]}\rec \setminus {[\Txt\G\Ex_C]}\rec \neq \emptyset\).\label{ap_th:c-w}
\end{theorem}

\begin{proof}
    We show this by using the Operator Recursion Theorem (\(\ORT\)) to provide a
    separating class of languages.
    To this end, let $h$ be
    the \(\Td\)-learner with \(h(\#)=\mbox{?}\) and, for all \(x,y\in \N\), let
    \(h(\langle x,y \rangle)= x\).
    Let \(\mathcal{L} = \Txt\Td\Ex_W(h)\cap \textbf{REC}\).
    Assume \(\mathcal{L}\) can be learned by a \(\Txt\G\Ex_C\)-learner \(h'\).
    By \citet{KP16}, we can assume \(h'\in\totalCp\). Then, by
    \ORT{} there exist indices \(e,p,q\in \N\)
    such that
    \begin{align*}
        L \coloneqq W_e &= \range(\varphi_p);\\
        \forall x\colon \tilde{T}(x) \coloneqq \varphi_p(x) &= \langle e, \varphi_q(\tilde{T}[x]) \rangle;\\
        \varphi_q(\varepsilon) &= 0;\\
        \forall \sigma\neq\varepsilon\colon \bar{\sigma} &= \min \{ \sigma' \subseteq \sigma \mid \varphi_q(\sigma') = \varphi_q(\sigma) \} ;\\
        \forall \sigma\neq\varepsilon\colon \varphi_q(\sigma) &= 
            \begin{cases}
                \varphi_q(\bar{\sigma}), &\text{if } \forall \sigma', \bar{\sigma} \subseteq \sigma' \subseteq \sigma \colon \Phi_{h'(\sigma')}(\langle
                    e,\varphi_q(\bar{\sigma})+1\rangle) > |\sigma|;\\
                \varphi_q(\bar{\sigma})+1, &\text{else, for min.\ } \sigma' \text{ contradicting the previous case, if } \\ 
                &\phantom{\sonstfalls.} \varphi_{h'(\sigma')}(\langle e,\varphi_q(\bar{\sigma})+1\rangle) = 0;\\
                \varphi_q(\bar{\sigma})+2, &\text{otherwise.}
            \end{cases}
    \end{align*}
    Here, \(\Phi\) is a Blum complexity
    measure, see~\citet{Blum67}. Intuitively, to define the next $\varphi_p(x)$, we add
    the same element to \(\content(\tilde{T})\) until
    we know whether \(\langle e, \tilde{T}[x]+1 \rangle \in
    C_{h'(\bar{\sigma})}\) holds
    or not. Then, we add the element contradicting this outcome.

    We first show that \(L\in \La\) and afterwards that \(L\) cannot be learned
    by \(h'\).
    To show the former, note that either \(L\) is finite or \(\tilde{T}\) is a non-decreasing
    unbounded computable enumeration of \(L\). Therefore, we have \(L\in
    \REC\). We now prove that $h$
    learns $L$. Let $T \in \Txt(L)$. For all $n \in \N$ where $T(n)$ is not the
    pause symbol, we have $h(T(n)) = e$. With \(n_0\in \N\) being minimal such that
    \(T(n_0)\neq\#\), we get for all \(n\geq n_0\) that
    \(\Td(h,T)(n)=e\). As \(e\) is a correct hypothesis, \(h\) learns
    \(L\) from \(T\) and thus we have that \(L\in\Txt\Td\Ex_W(h)\). Altogether, we
    get that \(L\in \La\).

    By assumption, \(h'\) learns \(L\) from the text \(\tilde{T} \in \Txt(L)\).
    Therefore, there exists \(n_0\in\N\) such that, for all \(n\geq
    n_0\), 
    \[
      h'(\tilde{T}[n]) = h'(\tilde{T}[n_0]) \text{ and } \chi_{L}=
      \varphi_{h'(\tilde{T}[n])},
    \]
    that is, \(h'(\tilde{T}[n])\) is a \(C\)-index for \(L\). Now, as $h'$
    outputs $C$-indices when converging, there are \(t, t' \geq n_0\) such that
    \[
      \Phi_{h'(\tilde{T}[t'])}(\langle e,\varphi_q(\tilde{T}[n_0])+1\rangle) \leq t.
    \]
    Let $t'_0$ and $t_0$ be the first such found. We show that
    $h'(\tilde{T}[t'_0])$ is no correct hypothesis of $L$ by distinguishing the
    following cases.
    \begin{enumerate}
      \itemin{1. Case:} $\varphi_{h'(\tilde{T}[t'_0])}(\langle
          e,\varphi_q(\tilde{T}[n_0])+1\rangle) = 0$. By definition of
          $\varphi_q$ and by minimality of $t_0'$, we have that $\langle
          e,\varphi_q(\tilde{T}[n_0])+1\rangle \in L$, however, the hypothesis
          of $h'(\tilde{T}[t'_0])$ says differently, a contradiction.
      \itemin{2. Case:} $\varphi_{h'(\tilde{T}[t'_0])}(\langle
          e,\varphi_q(\tilde{T}[n_0])+1\rangle) = 1$. By definition of $\varphi_q$ and
          by minimality of $t_0'$, we have that $\langle
          e,\varphi_q(\tilde{T}[n_0])+2\rangle \in L$, but $\langle
          e,\varphi_q(\tilde{T}[n_0])+2\rangle \notin L$. However, the hypothesis of
          $h'(\tilde{T}[t'_0])$ conjectures the latter to be in $L$, a contradiction.
        \qedhere 
    \end{enumerate} \noqed
\end{proof}

Furthermore, known equalities from learning $W$-indices directly apply in the
studied setting as well. These include the following.

\begin{theorem}[\citet{KS95}, \citet{Fulk90}]
    We have that \label{Thm:FulkKS}
    \[{[\Txt\It\Ex_W]}\rec
    \subseteq {[\Txt\Sd\Ex_W]}\rec \text{ and } {[\Txt\Psd\Ex_W]}\rec = {[\Txt\G\Ex_W]}\rec. \]
\end{theorem}

The remaining separations we will show in a more general way, see
Theorems~\ref{thm:sd-weakness} and~\ref{thm:it-sd-weak}. We continue by showing
that the latter result, namely that Gold-style learners may be assumed partially
set-driven, transfers to all considered cases. We generalize the result by
\citet{SchRicht84} and \citet{Fulk90}. The idea here is to, just as in the
$\Ex_W$-case, mimic the given learner and to search for minimal locking
sequences. Incorporating the result of \citet{KP16} that unrestricted Gold-style
learners may be assumed total, we even get a stronger result.

\begin{theorem}
   For \(\delta,\delta' \in \set{\Cind, \T}\), we have that\label{ap_th:g-psd}
   \[{[\tau(\delta)\Txt\G\delta'\Ex_C]}\rec =
    {[\tau(\delta)\totalCp\Txt\Psd\delta'\Ex_C]}\rec.\]
\end{theorem}

\begin{proof}
    We modify the proof as seen in~\citet{Fulk90}.
    The inclusion \({[\tau(\delta)\totalCp\Txt\Psd\delta'\Ex_C]}\rec \subseteq
    {[\tau(\delta)\Txt\G\delta'\Ex_C]}\rec\) follows immediately.
    For the other, let \(h\) be a
    \(\tau(\delta)\Txt\G\delta'\Ex_C\)-learner, which can
    assumed to be total by \citet{KP16} and let
    \(\mathcal{L} = \tau(\delta)\Txt\G\delta'\Ex_C(h)\cap\REC\). We define, for
    each finite set \(D \subseteq\N\) and \(t\in\N\),
    \begin{equation*}
        p(D, t) = \set{\sigma \in D^{\leq t} }{ \forall \tau \in D^{\leq t}\colon
            h(\sigma) = h(\sigma \tau)},
    \end{equation*}
    which, intuitively, contains potential locking sequences of \(h\). We define a
    \(\tau(\delta)\totalCp\Txt\Psd\delta'\Ex_C\)-learner \(h'\) for all finite sets
    \(D\) and \(t\in \N\) as
    \begin{equation*}
    h'(D, t) =
        \begin{cases}
              h(\min(p(D, t))), &\text{if } p(D, t) \neq \emptyset;\\
              \ind(\emptyset), &\text{otherwise.}
        \end{cases}
    \end{equation*}
    Note that \(h'\in \totalCp\) since \(h\in\totalCp\). To show that every
    language learned by \(h\) is also learned by \(h'\), let \(L \in
    \mathcal{L}\) and \(T \in \Txt(L)\). Let \(\sigma_0\) be a minimal
    locking sequence for \(h\) on \(L\). Let \(n_0\) be sufficiently large such
    that
    \begin{itemize}
        \item \(\content(\sigma_0) \subseteq \content(T[n_0])\),
        \item \(\abs{\sigma_0}\leq n_0\), and
        \item for all \(\sigma'\in{L}_\#^{*}\), with
            \(\sigma' < \sigma_0 \), there exists \(\tau\in
            {(\content(\sigma_0))}^{\leq n_0}_{\#}\) with \(h(\sigma')\neq
            h(\sigma'\tau)\).
    \end{itemize}

    Now, for
    all \(n \geq n_0\), we have \(\min(p(T[n], n)) = \sigma_0\) and, thus, \(h'\)
    outputs a correct hypothesis on \(T[n]\) which shows that \(L\in
    \tau(\delta)\totalCp\Txt\Psd\delta'\Ex_C(h')\).

    It remains to be shown that $h'$ preserves the restrictions imposed on $h$.
    This is clear whenever the restriction equals $\T$. For the remaining, we
    consider the following cases.
    \begin{enumerate}
      \itemin{1. Case:} $\delta=\Cind$. In this case, $h$ always outputs
          $C$-indices. Since $h'$ mimics $h$ or outputs $\ind(\emptyset)$, which
          also is an $C$-index, we have that $h'$ preserves $\delta$.
      \itemin{2. Case:} $\delta'=\Cind$. Let \(L\in\La\), \(T\in\Txt(L)\) and $n
          \in \N$. If $p(\content(T[n]), n) = \emptyset$, $h'(\content(T[n]),n)$
          outputs the $C$-index $\ind(\emptyset)$. Otherwise, if
          $p(\content(T[n]), n) \neq \emptyset$, let $\sigma =
          \min(p(\content(T[n]), n)) \in L^*$. Then, we have that
          $h'(\content(T[n]),n) = h(\sigma)$ which also is a $C$-index.\qedhere
    \end{enumerate} \noqed
\end{proof}

Also the former result of Theorem~\ref{Thm:FulkKS} holds in all considered
cases, as the same simulating argument (where one mimics the iterative learner
on ascending text with a pause symbol between two elements) suffices regardless
the exact setting. We provide the general result.

\begin{theorem}
    Let
    \(\delta,\delta'\in\set{\Cind,\True}\) and
    \(\mathcal{C}\in\set{\totalCp,\partialCp}\). Then, we have that\label{ap_thm:it_subset_sd}
    \begin{align*}{[\tau(\delta')\mathcal{C}\Txt\It\delta\Ex_C]}\rec\subseteq
    {[\tau(\delta')\mathcal{C}\Txt\Sd\delta\Ex_C]}\rec.\end{align*} 
\end{theorem}

\begin{proof}
    We adapt the proof of~\citet{KS95}.
    Let a \(h\) be a learner and let
    \(\La=\tau(\delta')\mathcal{C}\Txt\It\delta\Ex_C(h)\).
    We show that the following learner \(h'\)
    \(\tau(\delta')\mathcal{C}\Txt\Sd\delta\Ex_C\)-learns \(\La\).
    To that end, for any set \(D\),
    let \(\sort_\#(D)\) be the sequence of the elements in
    \(D\) sorted in ascending order, with a \(\#\) between each two elements,
    and let \(h^*\) be the starred form of \(h\).
    Now, we define \(h'\) as, for all finite sets \(D\),
    \begin{align*}
        h'(D)=
            \begin{cases}
                h^*(\sort_\#(D)),&\cIf h^*(\sort_\#(D))=h^*(\sort_\#(D)\concat\#);\\
                \ind(D),&\otw.
            \end{cases}
    \end{align*}
    Note that \(h'\) outputs a \(C\)-index, whenever \(h\) does so or when it
    outputs \(\ind\). Thus, \(h'\) preserves the \(\Cind\)-restrictions of
    \(h\). Moreover, if \(h\) is total, then so is \(h'\).
    To show that \(h'\) learns \(\La\), let \(L\in \La\). If \(L\) is
    finite, then either \(h^*(\sort_\#(L))=h^*(\sort_\#(L)\concat\#)\), in which
    case \(h\) converges to \(h'(L)=h^*(\sort_\#(L))\) on text
    \(\sort_\#(L)\concat\#^\infty\). Otherwise, we have \(h'(L)=\ind(L)\). In both
    cases, \(h'\) learns \(L\) as \(h'(L)\) is a correct  \(C\)-index for \(L\).

    On the other hand, if \(L\) is infinite, then \(h\) must converge to a
    \(C\)-index for \(L\) on the text \(\sort_\#(L)\). Let \(\sigma_0\)
    be the start sequence of \(\sort_\#(L)\) after which \(h\) is
    converged and let \(D_0=\content(\sigma_0)\). Then, for all \(x\in
    L\setminus D_0\), we have
    \(h^*({\sigma_0}\concat x)=h^*(\sigma_0)=h^*({\sigma_0}\concat\#)\)  as \(h\) is
    iterative. Therefore, for all \(D'\) with
    \(D_0 \subseteq D' \subseteq L \), we have
    \(h^*(\sort_\#(D'))=h^*(\sort_\#(D')\concat\#)\) and
    \(h^*(\sort_\#(D'))=h^*(\sort_\#(D_0))\), which is a correct hypothesis for
    \(L\). As \(h'(D')=h(\sort_\#(D'))\), we have convergence of \(h'\) to a
    correct \(C\)-index for \(L\) and thus \(h'\) learns \(L\).
\end{proof}

Interestingly, totality is no restriction solely for Gold-style (and due to the
equality also partially set-driven) learners. For the other considered learners
with restricted memory, being total lessens the learning capabilities. This
weakness results from the need to output some guess. A partial learner can await
this guess and outperform it. This way, we obtain self-learning languages
\citep{CK16} to show each of the three following separations.

\begin{theorem} We have that
  \({[\totalCp\Txt\Sd\Cind\Ex_C]}\rec \subsetneq {[\Txt\Sd\Cind\Ex_C]}\rec\).
\end{theorem}

\begin{proof}
    The inclusion \({[\totalCp\Txt\Sd\Cind\Ex_C]}\rec \subseteq
    {[\Txt\Sd\Cind\Ex_C]}\rec\) is straightforward. Suppose, by way of
    contradiction, that \( {[\totalCp\Txt\Sd\Cind\Ex_C]}\rec =
    {[\Txt\Sd\Cind\Ex_C]}\rec\). Let \(h\) be a learner
    such that, for all finite sets \(D \subseteq \N\)
    \begin{align*}
        h(D)=
            \begin{cases}
                \varphi_{\max(D)}(0), &\cIf D\neq\emptyset;\\
                \ind(\emptyset), &\otw.
            \end{cases}
    \end{align*}
    We now show that \(\mathcal{L} = \Txt\Sd\Cind\Ex_C(h)\cap\REC\) is a separating class
    contradicting the assumption that both classes are equally powerful. To that
    end,
    assume there exists a total learner \(h'\) with \(\mathcal{L}
    \subseteq \totalCp\Txt\Sd\Cind\Ex_C(h')\).
    By the Operator Recursion Theorem (\(\ORT\))  there exist an index \(e\in
    \N\), a strictly monotonically increasing function \(T \in \totalCp\) and
    \(c\in\totalCp\) such
    that, for all \(n, x\in \N\),
    \begin{align*}
        L  &= \range(T); \\
        \varphi_{e}&=\chi_{L};\\
        c(n)&=\content(T[n]); \\
        \varphi_{T(n)}(x) &= 
        \begin{cases}
            e,              &\text{if } \forall n' \leq n \colon h'(c(n'+1))
                \neq h'(c(n'+2));\\
            \ind(c(n+1)),   &\text{otherwise.}
        \end{cases}
    \end{align*}
    Note that there is a \(C\)-index for \(\range(T)\) because \(T\) is strictly
    monotonically increasing.
    Intuitively, if \(h'\) always makes mind changes on the start of the text
    \(T\), then \(\varphi_{T(n)}\) is a function that constantly outputs an index
    for a infinite set,  and otherwise, if \(h'\) repeats a hypothesis, then
    \(\varphi_{T(n)}\) is constantly an index for a
    finite set. 
    
    We now show that there exists a language that is learned by \(h\) but not by
    \(h'\).
    For this purpose, we consider the following cases.
    \begin{enumerate}
        \itemin{Case 1:} \(\forall n\colon h'(c(n+1)) \neq h'(c(n+2))\).
            In this case \(L \in \mathcal{L}\) holds because \(h\) will always
            output \(e\) on every sequence of a text for \(L\), which is a
            correct \(C\)-index for \(L\). But \(h'\) makes infinitely many mind
            changes on text $T$ and thus \(L
            \nsubseteq\totalCp\Txt\Sd\Cind\Ex_C(h')\).
    
        \itemin{Case 2:} \(\exists n\colon h'(c(n+1)) = h'(c(n+2))\).
            Let \(n_0\) be the smallest such \(n\). Then, \(h\) learns the
            languages \(c(n_0+1)\) and \(c(n_0+2)\) because the maximum of these
            sets is \(T(n_0)\) and \(T(n_0+1)\), respectively. Thus, \(h\) will
            output the correct hypothesis \(\ind(c(n_0+1))\) or
            \(\ind(c(n_0+2))\), respectively. But \(h'\) cannot differentiate
            between those two different languages. Thus, it learn both
            simultaneously.  Therefore, we again have \(\La \nsubseteq
            \totalCp\Txt\Sd\Cind\Ex_C(h')\).\qedhere
    \end{enumerate} \noqed
\end{proof}

\begin{theorem}
  We have that \({[\totalCp\Txt\It\Cind\Ex_C]}\rec \subsetneq
  {[\Txt\It\Cind\Ex_C]}\rec\).
\end{theorem}

\begin{proof}
  The inclusion \({[\totalCp\Txt\It\Cind\Ex_C]}\rec \subseteq
  {[\Txt\It\Cind\Ex_C]}\rec\) follows immediately. We prove that we have a proper
  inclusion by providing a separating class using the Operator Recursion
  Theorem (\(\ORT\)). Suppose now, by way of contradiction,
  that \( {[\totalCp\Txt\It\Cind\Ex_C]}\rec = {[\Txt\It\Cind\Ex_C]}\rec\). Let \(h\) be
  a \(\Txt\It\Cind\Ex_C\)-learner such that \(h(\varepsilon)=\pad(0,0)\) and,
  for all \(e, k, x\in \N\),
  \begin{equation*}
    h(\pad(e, k), x) = 
        \begin{cases}
            \divs,                       &\text{if } \varphi_x(0)\divs;\\
            \pad(e,k),                      &\text{if } k > \pi_2(\varphi_x(0));\\
            \pad(\pi_1(\varphi_x(0)),\pi_2(\varphi_x(0))),    &\text{otherwise.}
        \end{cases}
    \end{equation*}
    Recall that \(\pi_1,\pi_2\) are the inverse functions to the pairing
    function \(\langle \cdot,\cdot \rangle \).
    Intuitively, \(h\) interprets each datum \(x\) as the index of a function and
    outputs the first component of \(\varphi_{x}(0)\) where the second component
    of  \(\varphi_{x}(0)\) is maximal.
    Now, let \(\mathcal{L} = \Txt\It\Cind\Ex_C(h)\cap\REC\). By our assumption there is a
    \(\totalCp\Txt\It\Cind\Ex_C\)-learner \(h'\) that learns \(\mathcal{L}\).
    For notational convenience, we use the starred learner \({(h')}^{*}\). With
    the \(\ORT\) there exist an index \(e\in \N\)
    and a strictly monotonically increasing \(T \in
    \totalCp\) such that, for all \(n, x\in \N\),
        \begin{align*}
            L \coloneqq C_e &= \range(T); \\
            \varphi_{T(n)}(x) &= 
                \begin{cases}
                    \langle e, 0 \rangle, &\begin{aligned}
                        \text{if } \forall n' \leq n\colon {(h')}^*(T[n']) \neq
                        {(h')}^*(T[n'+1])\ \lor \\ {(h')}^*(T[n'+1]) \neq
                        {(h')}^*(T[n'+2]);
                        \end{aligned}\\
                    \langle \ind(\content(T[n+1])), n \rangle,&\text{otherwise.}
                \end{cases}
        \end{align*}
    Note that we can find a \(C\)-index for \(\range(T)\) because \(T\) is strictly
    monotonically increasing.
    We now consider the following cases.
    \begin{itemize}
        \itemin{Case 1:} \(\forall n\in \N\colon {(h')}^*(T[n]) \neq
            {(h')}^*(T[n+1])\ \lor\ {(h')}^*(T[n+1]) \neq {(h')}^*(T[n+2])\). On
            any element $x \in L$, $h$ outputs \(\pad(e,0)\), which is a correct
            \(C\)-Index for \(L\). Thus, once $h$ sees the first non-pause
            symbol, it converges correctly and, thus, \(L\in \La\).  But \(h'\)
            makes infinitely many mind changes on text $T$ and thus cannot learn
            \(\mathcal{L}\).
      
        \itemin{Case 2:} \(\exists n\in \N\colon {(h')}^*(T[n]) =
            {(h')}^*(T[n+1])\ \land\ {(h')}^*(T[n+1]) = {(h')}^*(T[n+2])\).  Let
            \(n_0\) be the smallest such \(n\). Then, \(h\) learns the finite
            languages \(\content(T[n_0+1])\) and \(\content(T[n_0+2])\) because
            \(\varphi_{T(n_0)}(0)\) and \(\varphi_{T(n_0+1)}(0)\) have the
            maximum second component in the respective set and, thus, \(h\)
            converges to \(\pad(\ind(\content(T[n_0+1])), n_0)\) and
            \(\pad(\ind(\content(T[n_0+2])), n_0+1)\), respectively. But by the
            assumption of this case, \(h'\) converges to same hypothesis on the
            texts \(T[n_0] \concat T{(n_0)}^{\infty}\) and \(T[n_0 + 1] \concat
            T{(n_0 + 1)}^{\infty}\), which are texts of different languages.
            Thus, \(h\) cannot learn \(\mathcal{L}\).\qedhere
    \end{itemize} \noqed
\end{proof}

\begin{theorem}
  We have that ${[\totalCp\Txt\Td\Cind\Ex_C]}\rec \subsetneq
  {[\Txt\Td\Cind\Ex_C]}\rec$.
\end{theorem}

\begin{proof}
    The inclusion ${[\totalCp\Txt\Td\Cind\Ex_C]}\rec \subseteq
    {[\Txt\Td\Cind\Ex_C]}\rec$ follows immediately. To prove that the inclusion is proper,
    we provide a separating class using the \tORT. Let \(h\) be a
    \(\Td\)-learner with \(h(\#)=\mbox{?}\) and, for all \(x\in \N\),
    \(h(x)=\varphi_{x}(0)\). Let \(\La=\Txt\Td\Cind\Ex_C(h)\). Now, assume there
    exists a learner \(h'\) with  \(\La \subseteq \totalCp\Txt\Td\Cind\Ex_C(h')\).
    Then, with \ORT{} there exists \(a\in \totalCp\) such that for all \(x,n\in
    \N\)
    \begin{align*}
        \varphi_{a(n)}(x)= 
        \begin{cases}
            \ind(\set{a(0),a(1)}), &\cIf h'(a(0))\neq h'(a(1));\\
            \ind(\set{a(n)}), &\otw.
        \end{cases}
    \end{align*}
    Intuitively, if \(h'\) suggests different hypotheses for \(a(0)\) and \(a(1)\)
    then both are in the same language and vice versa. We now show that in both
    cases, there is a language learned by \(h\) which cannot be learned by
    \(h'\). We distinguish the following cases.
    \begin{itemize}
        \itemin{Case 1:} \(h'(a(0))\neq h'(a(1))\). Then, we have \(
            \set{a(0),a(1)}\in
            \La\), as \(h\) outputs a \(C\)-Index for this set on both elements
            of the set. But \(h'\) does not converge on the text
            \({(a(0)a(1))}^\infty\) and thus cannot learn this set.
        \itemin{Case 2:} \(h'(a(0))= h'(a(1))\). Then, by construction, we have
            \( \set{a(0)}, \set{a(1)} \in \La\). But \(h'\) suggests the same
            hypothesis on \({a(0)}^\infty\) and \({a(1)}^\infty\) and thus can
            learn at most one of these two sets. \qedhere
    \end{itemize} \noqed
\end{proof}

Next, we show the gradual decrease of learning power the more we require the
learners to output characteristic indices. We have already seen in
Theorem~\ref{ap_th:c-w} that converging to $C$-indices lessens learning power.
However, this allows for more learning power than outputting these indices
during the whole learning process as shows the next theorem. The idea is that
such learners have to be certain about their guesses as these are indices of
characteristic functions. When constructing a separating class using
self-learning languages \citep{CK16}, one forces the $\Cind$-learner to output
$C$-indices on certain languages to, then, contradict its choice there. This
way, the $\Ex_C$-learner learns languages the $\Cind$-learner cannot. The
following theorem holds.

\begin{theorem}
  We have that \({[\Txt\It\Ex_C]}\rec \setminus {[\Txt\G\Cind\Bc_C]}\rec \neq \emptyset\).\label{ap_th:cind}
\end{theorem}

\begin{proof}
  We prove this by contradiction by providing a class of languages in
  ${[\Txt\It\Ex_C]}\rec$ which is not in ${[\Txt\G\Cind\Bc_C]}\rec$. Let $h$ be
  the following $\It$-learner. Let $p_\N$ be an index for the set of all natural
  numbers. For any $e, x \in \N$, we define
  \begin{align*}
    h(\varepsilon) &= \ind(\emptyset); \\
    h(e,x) &= 
    \begin{cases}
        e, &\falls \pi_2(e) = 1 \wedge \pi_2(x) = 1 \wedge \pi_1(x) < \pi_1(e); \\
        \langle \pi_1(x), 1\rangle, &\sonstfalls \pi_2(e) \neq 1 \wedge \pi_2(x) = 1; \\
        \langle \pi_1(x), 2\rangle, &\sonstfalls \pi_2(e) = 0 \vee (\pi_2(x)=2 \wedge \pi_1(x) < \pi_1(e)); \\
        e, &\sonst.
    \end{cases}
  \end{align*}
  Without loss of generality, we may assume that $\ind(\emptyset)  = \langle 0,0
  \rangle$. This way, we can distinguish whether it was the previous hypothesis
  or not. Intuitively, while $h$ only sees elements with second component two,
  it outputs the minimal $\langle \pi_1(x), 2 \rangle$ it has seen. Once it sees
  an element with second component one, it outputs the coded tuple $\langle
  \pi_1(x), 1\rangle$, which, if no other such elements are presented, is its
  final hypothesis. Otherwise, $h$ outputs the minimal $\langle \pi_1(x), 1 \rangle$.
  Now, let
  \(\mathcal{L} = \Txt\It\Ex_C(h) \cap \textbf{REC}\) and assume there exists a
  learner $h'$ which $\Txt\G\Cind\Bc_C$-learns $\La$, that is, $\La \subseteq
  \Txt\G\Cind\Bc_C(h')$. By the Operator Recursion Theorem (\ORT), there exist
  total computable increasing
  functions $a, \tilde{a} \in \totalCp$ and indices $e, p \in \N$
  such that for all $n, x \in \N$
  \begin{align*}
    \tilde{a}(x) &= \langle a(x), 2 \rangle; \\
    L_n &\coloneqq \content(\tilde{a}[n]) \cup \{ \langle a(n), 1 \rangle \}; \\
    L \coloneqq C_e &= \range(\varphi_p); \\
    T(x) \coloneqq \varphi_p(x) &=
        \begin{cases}
          \langle a(2x), 2 \rangle , &\falls \varphi_{h'(\varphi_p[x])}(\langle a(2x), 2 \rangle) = 0; \\
          \langle a(2x+1), 2 \rangle, &\sonstfalls \varphi_{h'(\varphi_p[x])}(\langle a(2x), 2 \rangle) = 1; \\
          \divs, &\sonst.
        \end{cases} \\
    \varphi_{\langle a(n), 2 \rangle}(x) &= \varphi_e(x) = 
        \begin{cases} 
          1, &\falls \langle a(2x), 2 \rangle \in \content(T[x+1]);\\
          0, &\sonstfalls \langle a(2x+1), 2 \rangle  \in \content(T[x+1]); \\
          \divs, &\sonst.
        \end{cases} \\
    \varphi_{\langle a(n), 1 \rangle}(x) &= \varphi_{\ind(L_n)}(x) = \chi_{L_n}(x);
  \end{align*}
  First, note that, for any $n \in \N$, $h$ learns $L_n$ as 
  it eventually outputs $\langle a(n), 1 \rangle$, a $C$-index for $L_n$, and never changes
  its mind again. As $h'$ learns these as well, it outputs a $C$-index on every initial
  sequence of elements in $\range(\tilde{a})$. Thus,
  $\varphi_p$ is total and there exists a $C$-index $e$ for its range. We now show, that $h$ learns the decidable language $L$,
  while $h'$ does not. As for any $x \in L$ there exists $n \in \N$ such that we
  have $x = \langle a(n), 2 \rangle$ and $\varphi_{x} = \varphi_{\langle
  a(n), 2 \rangle} = \varphi_e$, we have that $h$ identifies $L$ correctly once
  it sees the minimal such element in $L$. On the other hand, we show that $h'$ cannot learn $L$ from text $T$. Let $x \in \N$ and consider the following cases.
  \begin{enumerate}
    \itemin{1. Case:} $\varphi_{h'(\varphi_p[x])}(\langle a(2x), 2 \rangle) = 0$. Thus, $\langle a(2x), 2 \rangle$ is not in the hypothesis of $h'$, but it is in $L$.
    \itemin{2. Case:} $\varphi_{h'(\varphi_p[x])}(\langle a(2x), 2 \rangle) =
        1$. Here, $\langle a(2x), 2 \rangle$ is in the hypothesis of $h'$, but, as $a$
        is strictly monotonically increasing, it is not in $L$.  
  \end{enumerate} 
  Thus, none of the hypothesis $h'(T[x])$ identifies $L$ correctly.
\end{proof}

Since languages which can be learned by iterative learners can also be learned
by set-driven ones (see Theorem~\ref{ap_thm:it_subset_sd}), this result
suffices. Note that the idea above requires some knowledge on previous elements.
Thus, it is no coincidence that this separation does not include transductive
learners. Since these learners base their guesses on single elements, they
cannot see how far in the learning process they are. Thus, they are forced to
always output $C$-indices. The following theorem holds.

\begin{theorem}
  We have that \({[\Txt\Td\Cind\Ex_C]}\rec = {[\Txt\Td\Ex_C]}\rec\).
\end{theorem}

\begin{proof}
    The inclusion \({[\Txt\Td\Cind\Ex_C]}\rec \subseteq  {[\Txt\Td\Ex_C]}\rec\) is
    immediate. For the other, let \(h\) be a \(\Txt\Td\Ex_C\)-learner
    and \(\La=\Txt\Td\Ex_C(h)\cap\REC\). We show that \(h\) is, in particular, a
    \(\Cind\)-learner, i.e., \(\La=\Txt\Td\Cind\Ex_C(h)\) 
    holds as well. Assume the contrary, that is, \(\La\neq\Txt\Td\Cind\Ex_C(h)\). Then
    there exists a \(L\in\mathcal{L}\) and a \(x\in L\) such that
    \(h(x)\) is no \(C\)-index. Now, given any text \(T \in \Txt(L)\), consider
    the text, for all $n \in \N$,
    \[
      T'(n) = \begin{cases} T(n), &\falls n \text{ is even}, \\ x, &\sonst. \end{cases}
    \]
    This text of the language $L$ contains infinitely many occurrences of $x$
    and, therefore, the $\Td$-learner $h$ cannot converge to a $C$-index on this
    text.
\end{proof}

For the remainder of this section, we focus on learners which output
characteristic indices on \emph{arbitrary} input, that is, we focus on
$\tau(\Cind)$-learners. First, we show that the requirement of always outputting
$C$-indices lessens a learners learning power, even when compared to total
$\Cind$-learners. To provide the separating class of self-learning languages,
one again awaits the $\tau(\Cind)$-learner's decision and then, based on these,
learns languages this learner cannot. The following result holds.

\begin{theorem}
  We have that \({[\totalCp\Txt\Td\Cind\Ex_C]}\rec \setminus
    {[\tau(\Cind)\Txt\G\Bc_C]}\rec \neq \emptyset\).\label{ap_th:tau-cind} 
\end{theorem}

\begin{proof}
    We prove the result by providing a separating class of languages. Let $h$ be
    the \(\Td\)-learner with \(h(\#)=\mbox{?}\) and, for all \(x,y\in \N\), let
    \(h(\langle x,y \rangle)= x\). By construction, \(h\) is total and
    computable. Let~$\mathcal{L} = \totalCp\Txt\Td\Cind\Ex_C(h) \cap \REC$. We
    show that there is no $\tau(\Cind)\Txt\G\Bc_C$-learner learning $\La$ by way
    of contradiction.
    Assume there is a \(\tau(\Cind)\Txt\G\Bc_C\)-learner \(h'\) which learns
    $\mathcal{L}$. With the Operator Recursion Theorem (\(\ORT\)), there are
    \(e,p \in \N\) such that for all
    \(x\in \N\)
    \begin{align*}
        L&\coloneqq \range(\varphi_p);\\
        \varphi_{e}&=\chi_{L};\\
        \tilde{T}(x) \coloneqq \varphi_p(x) &=
            \begin{cases}
                \langle e,2x \rangle, &\text{if } \varphi_{h'(\varphi_p[x])}(\langle
                e, 2x \rangle) = 0; \\ 
                \langle e,2x+1 \rangle, &\text{otherwise.}
            \end{cases}
    \end{align*} 
    Intuitively, for all \(x\) either \(\varphi_{p}(x)\) is an element of \(L\)
    if it is not in the
    hypothesis of \(h'\) after seeing~\(\varphi_{p}[x]\), or there is an element
    in this hypothesis that is not in \(\content(\tilde{T})\). As any hypothesis of $h'$ is
    a $C$-index, we have that $\varphi_p \in \totalCp$ and, as \( \varphi_{p}\)
    is strictly monotonically increasing, that $L$ is decidable.

    We now prove that \(L \in \mathcal{L}\) and afterwards that \(L\) cannot be
    learned by \(h'\). First, we need to prove that $h$
    learns $L$. Let $T \in \Txt(L)$. For all $n \in \N$ where $T(n)$ is not the
    pause symbol, we have $h(T(n)) = e$. Let \(n_0\in \N\) with
    \(T(n_0)\neq\#\). Then, we have, for all \(n\geq n_0\), that
    \(\Td(h,T)(n)=e\) and, since \(e\) is a hypothesis of $L$, \(h\) learns
    \(L\) from \(T\). Thus, we have that \(L\in\totalCp\Txt\Td\Cind\Ex_C(h) \cap \REC\).

    By assumption, \(h'\) learns \(\mathcal{L}\) and thus it also needs to learn
    $L$ on text $\tilde{T}$. Hence, there is $x_0$ such that for all \(x \geq x_0\) the 
    hypothesis $h'(\tilde{T}[x])=h'(\varphi_{p}[x])$  is a \(C\)-index for \(L\). We now consider
    the following cases.

    \begin{enumerate}
      \itemin{1. Case:} \(\varphi_{h'(\varphi_p[x])}(\langle e, 2x \rangle) =
          0\). By construction, we have that \(\tilde{T}(x) = \langle e, 2x
          \rangle\). Therefore, \(\langle e, 2x \rangle \in L\), which
          contradicts \(h'(\varphi_p[x])\) being a correct hypothesis.
      \itemin{2. Case:} \(\varphi_{h'(\varphi_p[x])}(\langle e, 2x \rangle) =
          1\). By construction, we have that \(\tilde{T}(x) \neq \langle e, 2x
          \rangle\) and thus, because \(\tilde{T}\) is strictly monotonically
          increasing, \(\langle e, 2x \rangle \notin L=\content(\tilde{T})\).
          This, again, contradicts \(h'(\varphi_p[x])\) being a correct
          hypothesis.
    \end{enumerate}
    As in all cases \(h'(\varphi_{p}[x])\) is a wrong hypothesis, $h'$ cannot
    learn $\mathcal{L}$.
\end{proof}

It remains to be shown that memory restrictions are severe for such learners as
well. First, we show that partially set-driven learners are more powerful than
set-driven ones. As witnessed originally by \citet{SchRicht84} and
\citet{Fulk90} (for $W$-indices), this is solely due to the lack of learning
time. We provide the following theorem. We already separate from behaviorally
correct learners, as we will need this stronger version later on.

\begin{theorem}%
    We have that \({[\tau(\Cind)\Txt\Psd\Ex_C]}\rec \setminus
        {[\Txt\Sd\Bc_W]}\rec\neq\emptyset\).\label{thm:sd-weakness}
\end{theorem}

\begin{proof}
    We prove the theorem by providing a separating class \(\La\). For all
    \(e\in\N\), we define
    \begin{align*}
        L_e&=\set{\langle e,x \rangle }{x\in\N};\\
        L'_e&=\set{\langle e,x \rangle }{\varphi_{e}(0)\convs \wedge x\leq \varphi_{e}(0)
        };\\
        \La&=\bigcup_{e \in \N}( \set{L_e}{\varphi_{e}(0)\divs }\cup
        \set{L'_e}{\varphi_{e}(0)\convs}).
    \end{align*}
    Note that \(\La\subseteq \REC\).
    First, we provide a learner \(h \) such that
    \(\La \subseteq \tau(\Cind)\Txt\Psd\Ex_C(h)\cap\REC\). To define \(h\), we
    need the following auxiliary functions. Due to the S-m-n Theorem there
    exist \(f,p,p'\in\totalCp\) such that for all finite sets \(D\) and  all
    \(e,x\in\N\)
    \begin{align*}
        f(D) &= 
            \begin{cases}
                \pi_1(\min(D)), &\cIf D\neq\emptyset;\\
                0, &\otw;
            \end{cases}\\
        \varphi_{p(e)}&=\chi_{L_e};\\
        \varphi_{p'(e,x)}&=\chi_{\set{\langle e,y \rangle }{y\leq x }}.
    \end{align*}
    Intuitively, we use $f$ to recover the first component of the minimal given
    element.  With \(p\) and \(p'\) we can generate \(C\)-Indices for \(L_e\)
    and \(L'_e\), respectively. Now, we define the learner $h$ as, for all
    finite sets $D$ and all $t \in \N$,
    \begin{align*}
        h(D,t)=
            \begin{cases}
                \ind(\emptyset), &\cIf D=\emptyset;\\
                p(f(D)), &\sonstfalls  \Phi_{f(D)}(0)> t;\\
                p'(f(D), \varphi_{f(D)}(0)), &\otw.
            \end{cases}
    \end{align*}
    Intuitively, given elements of the form $\langle e, x \rangle$, \(h\) suggests
    \(L_e\) until it witnesses \(\varphi_{e}(0)\convs\), whereupon it
    suggests \(L'_e\). Note that \(h\) is a \(\tau(\Cind)\)-learner by
    construction. 

    To show that \(\La \subseteq \tau(\Cind)\Txt\Psd\Ex_C(h)\), let
    \(e\in\N\). If \(\varphi_{e}(0)\divs\), \(h\) needs to learn \(L_e\). After
    seeing the first non-pause symbol, \(h\) constantly outputs \(p(e)\), which
    is a correct index for \(L_e\). If, otherwise, \(\varphi_{e}(0)\convs\),
    \(h\) needs to learn \(L'_e\). Let \(T\in\Txt(L'_e)\) and  \(n_0\in\N\)
    big enough such that \(T[n_0]\neq\emptyset\) and \(n_0\geq\Phi_e(0)\). Then
    for all \(n\geq n_0\) we have \(h(T[n],n)=p'(e,\varphi_{e}(0))\) and thus
    \(h\) learns \(L'_e\) as well.

    It remains to be shown that there is no learner \(h'\) such that \(\La \subseteq
    \Txt\Sd\Bc_W(h')\). Assume the opposite, i.e., let \(h'\) be a learner with
    \(\La \subseteq \Txt\Sd\Bc_W(h')\). By Kleenes Recursion Theorem there exists an
    index \(e\in\N\) such that, for all $x \in \N$,
    \[
        \varphi_e(x) = \begin{cases}
            m, &\falls \exists m \colon \langle e,m+1 \rangle\in
                C_{h'(\set{\langle e,x  \rangle}{x\leq m })}; \\
            \divs, &\sonst.
        \end{cases}
    \]
    If ever, we take the first such $m$ found. We differentiate whether
    \(\varphi_{e}(0)\convs\) or not.
    \begin{itemize}
        \itemin{Case 1:}\(\varphi_{e}(0)\convs\). Then \(h'\) has to learn
            \(L'_e\). Let  \(m=\varphi_{e}(0)\). By definition of \(e\) we have
            \(\langle e,m+1 \rangle \in C_{h'(L'_e)}\). As \(\langle e,m+1
            \rangle\notin L'_e \), this contradicts \(h'\) learning \(L_e'\).
        \itemin{Case 2:}\(\varphi_{e}(0)\divs\). Then \(h'\) has to learn
            \(L_e\). Let \(T \in \Txt(L_e)\) be the text with, for all \(i\in\N\),
            \(T(i)=\langle e,i \rangle \). By definition of \(e\) we have, for all \(m\in\N\),
            \[\langle e,m+1 \rangle \notin C_{h'(\set{\langle e,x \rangle }{x\leq m })} = C_{h'(\content(T[m+1]))}.\]
            Therefore, \(h'\) cannot converge to a correct hypothesis for
            \(L_e\) on \(T\) and, thus, not learn it.\qedhere
    \end{itemize}\noqed
\end{proof}

In turn, this lack of time is not as severe as lack of memory. The standard
class (of recursive languages) to separate set-driven learners from iterative
ones \citep{JORS99} can be transferred to the setting studied in this paper. We
obtain the following result. 

\begin{theorem}
    We have that
    \({[\tau(\Cind)\Txt\Sd\Ex_C]}\rec \setminus
    {[\Txt\It\Ex_W]}\rec \neq \emptyset\). \label{thm:it-sd-weak}
\end{theorem}

\begin{proof}
    This is a standard proof and we include it for completeness \citep{JORS99}.
    We show this theorem by stating a class of languages that can be learned
    by a \(\tau(\Cind)\Txt\Sd\Ex_C\)-learner, but any
    \(\Txt\It\Ex_W\)-learner fails to do so. To that end, let \(\mathcal{L} =
    \{D \cup \{0 \} \mid D \subseteq_\Fin \mathbb{N}\} \cup
    \{\mathbb{N}^+\}\).
    We define the \(\Sd\)-learner \(h\) for all finite sets \(D\), with
    \(p\) being a \(C\)-Index for \(\N^+\), as
    \begin{equation*}
        h(D) = \begin{cases}
            \ind(D), &\text{if } 0 \in D;\\
            p, &\text{otherwise.}
        \end{cases}
    \end{equation*}
    It is easy to verify that 
    \(\La \subseteq \tau(\Cind)\Txt\Sd\Ex_C(h)\). Now, assume
    there is a \(\Txt\It\Ex_W\)-learner \(h'\) that learns
    \(\mathcal{L}\) and let \(\sigma\) be a locking sequence of \(h'\) on
    \(\mathbb{N}^+\) with \(x = \max(\content(\sigma))\).
    The texts \(\sigma \concat (x + 1) \concat 0^\infty\) and \(\sigma
    \concat (x + 2) \concat 0^\infty\) are texts for distinct languages from
    \(\mathcal{L}\) but
    \(h'\) suggests exactly the same hypotheses on both texts and
    can therefore not be \(\Ex_W\)-successful on both languages.
\end{proof}

Lastly, we show that transductive learners, having basically no memory, do
severely lack learning power. As they have to infer their conjectures from
single elements they, in fact, cannot even learn basic classes such as $\{ \{0
\}, \{1 \}, \{0,1\}\}$. The following result holds. It concludes the map shown
in Figure~\ref{Map:DifferentSettings} and, therefore, also this section.

\begin{theorem}
    For \(\beta\in \set{\It,\Sd}\), we have that \label{thm:it-td-weak}
    \[{[\tau(\Cind)\Txt\beta\Ex_C]}\rec
    \setminus{[\Txt\Td\Ex_W]}\rec\neq\emptyset.\]
\end{theorem}

\begin{proof}
        We include this standard proof for completeness. We follow
        \citet{CCJS07} and show that $\La = \{ \{0\}, \{1\}, \{0,1\}\}$ is a
        separating class. Immediate, we have that $\La$ can be learned by a
        $\tau(\Cind)\Txt\beta\Ex_C$-learner and, thus, $\La \in
        [\tau(\Cind)\Txt\beta\Ex_C]\rec$. Now, assume there exists a learner
        $h'$ $\Txt\Td\Ex_W$-learning $\La$. Consider the texts $T_0 = 0^\infty
        \in \Txt(\{0\})$ and $T_1 = 1^\infty \in \Txt(\{1\})$. As $h'$ must
        identify both languages on their respective text, we have that, for $x
        \in \{0,1\}$, $h'(x)$ must be a $C$-index for $\{ x \}$. However, then
        $h'$ cannot output a $C$-index of $\{0,1\}$ on the text $T = 0 \concat
        1^\infty$, a contradiction. 
\end{proof}

\section{Syntactic versus Semantic Convergence to $C$-indices}\label{Sec:C-Convergence}

In this section we investigate the effects on learners when we require them to
converge to characteristic indices. We study both syntactically converging
learners as well as semantically converging ones. In particular, we compare
learners imposed with different well-studied memory restrictions. 

Surprisingly, we observe that, although $C$-indices incorporate and, thus,
require the learner to obtain more information during the learning process than
$W$-indices, the relative relations of the considered restrictions remain the
same. We start by gathering results which directly follow from the previous
section. In particular, the following corollary holds.

\begin{corollary}
    We have that 
    \begin{align*}
        {[\Txt\Psd\Ex_C]}\rec &= {[\Txt\G\Ex_C]}\rec, (\text{Theorem~\ref{ap_th:g-psd}}), \\
        {[\Txt\It\Ex_C]}\rec &\subseteq {[\Txt\Sd\Ex_C]}\rec, (\text{Theorem~\ref{ap_thm:it_subset_sd}}), \\
        {[\Txt\G\Ex_C]}\rec &\setminus {[\Txt\Sd\Bc_C]}\rec \neq \emptyset, (\text{Theorem~\ref{thm:sd-weakness}}), \\
        {[\Txt\Sd\Ex_C]}\rec &\setminus {[\Txt\It\Ex_C]}\rec \neq \emptyset, (\text{Theorem~\ref{thm:it-sd-weak}}), \\
        {[\Txt\It\Ex_C]}\rec &\setminus {[\Txt\Td\Ex_C]}\rec \neq \emptyset, (\text{Theorem~\ref{thm:it-td-weak}}).
    \end{align*}
\end{corollary}

We show the remaining results. First, we show that, just as for $W$-indices,
behaviorally correct learners are more powerful than explanatory ones. We
provide a separating class exploiting that explanatory learners must converge to
a single, correct hypothesis. We collect elements on which mind changes are
witnessed, while maintaining decidability of the obtained language. The
following result holds.

\begin{theorem}
  We have that ${[\Txt\Sd\Bc_C]}\rec \setminus {[\Txt\G\Ex_C]}\rec \neq \emptyset$.
\end{theorem}

\begin{proof}
  In order to provide a separating class of languages, we consider the learner,
  for all finite $D \subseteq \N$, $h(D) = \max(D)$. Let $\La =
  \Txt\Sd\Bc_C(h)\cap\REC$. We show that there exists no learner $h'$ that
  $\Txt\G\Ex_C$-learns $\La$. To that end, assume there exists such a learner
  $h'$, that is, $\La \subseteq \Txt\G\Ex_C(h')$. Without loss of generality, we
  may assume $h'$ to be total, as is shown in \citet{KP16}.
  
  Using the Operator Recursion Theorem (\ORT), there exist an \emph{interleaved
  increasing}
  \footnote{A function $a$ is called \emph{interleaved increasing}
    if, for all $n$, we have $a(0,n) < a(1,n) < a(0, n+1)$.} function $a \in
  \totalCp$, a sequence of sequences ${(\sigma_j)}_{j \in \N}$ and functions $f,
  i_0, s \in \partialCp$ such that, for all $i, j, k, t, x \in \NAT$ and $b \in
  \{0,1\}$, we have
  \begin{align*}
  P_j(t) &\Leftrightarrow h'({\sigma_j} \concat a{(0,|\sigma_j|)}^t) \neq
    h'(\sigma_j) \lor h'({\sigma_j} \concat a{(1,|\sigma_j|)}^t) \neq
    h'(\sigma_j); \\
  s(j) &= \mu t.P_j(t);\\
  \sigma_0 &= \varepsilon;  \\
  \sigma_{j+1} &= 
    \begin{cases}
      \divs,    &\falls s(j) \divs; \\
      {\sigma_j}\concat a{(0,|\sigma_j|)}^{s(j)}, &\sonstfalls h'({\sigma_j}\concat a{(0,|\sigma_j|)}^{s(j)}) \neq h'(\sigma_j);\\
      {\sigma_j}\concat a{(1,|\sigma_j|)}^{s(j)}, &\sonst.
    \end{cases}\\
  \varphi_{a(b,i)}(x) &= 
  \begin{cases}
    1, &\falls x = a(b,i); \\
    0, &\sonstfalls x = a(1-b,i); \\
    f(b',k'), &\sonstfalls \exists k' \in \N \ \exists b' \in \settwo{0,1} \colon x = a(b',k');\\
    0, &\sonst.
  \end{cases} \\
  i_0(k) &= \max\{ j \mid \sigma_j\convs \wedge |\sigma_j| \leq k \}; \\
  f(b,k) &= 
  \begin{cases}
    0, &\falls k > |\sigma_{i_0(k)}| ;\\
    1, &\sonstfalls s(i_0(k)) \convs \AND a(b,k) \in \content(\sigma_{i_0(b,k)+1});\\
    0, &\sonstfalls s(i_0(k)) \convs; \\
    \divs, &\sonst.
  \end{cases}
  \end{align*}
  Note that $b'$ and $k'$ in the third case of $\varphi_{a(b,i)}(x)$ are, if
  they exist, unique as $a$ is interleaved increasing. The intuition is the
  following. For $j \in \N$, we extend the sequence $\sigma_j$ as soon as $h'$
  makes a particular mind change, if ever. This guarantees that $h'$ cannot
  learn certain languages $h$ can. Furthermore, for suitable $b,i \in \N$, every
  element $a(b,i)$ of the sequence encodes the language $\bigcup_{j \in \N,
  \sigma_j \convs} \content(\sigma_j)$ (as $C$-index). This encoding is done
  using function $f$ which, given the right circumstances, can decide whether an
  element belongs to the mentioned language or not. We first provide a proof for
  this claim. 
  \begin{claim}\label{Claim:fdecide}
    Let $b \in \settwo{0,1}$ and $j_0$ such that $\sigma_{j_0}$ is defined. Let
    $k = |\sigma_{j_0}|$. Then, if $\sigma_{j_0 + 1}$ is defined, $f(b,k)$
    correctly decides whether $a(b,k) \in \bigcup_{j \in \N, \sigma_j \convs}
    \content(\sigma_j)$. 
  \end{claim}
  \begin{proof}
    Let $\sigma_{j_0+1}$ be defined. Then, $j_0 = i_0(k)$ and $s(j_0) \convs$
    and we have that, by definition, $f$ correctly decides whether $a(b,k) \in
    \bigcup_{j \in \N, \sigma_j \convs} \content(\sigma_j)$. \claimqedhere\noqed
  \end{proof}

  We show that there exists a language $h$ can learn, but $h'$ cannot. To that end, we distinguish the following cases. 
  \begin{enumerate}
    \itemin{1. Case:} For all $j\in\N$, $\sigma_j$ is defined. Let $\tilde{T} =
        \bigcup_{j \in \N} \sigma_j$ and let $L = \content(\tilde{T})$. We first
        show that $h$ learns $L$. Let $T \in \Txt(L)$ and let $n_0$ be minimal
        such that $\content(T[n_0]) \neq \emptyset$. Let $n \geq n_0$ and $D
        \coloneqq \content(T[n])$. Furthermore, let $b \in \settwo{0,1}$ and $i
        \in \N$ be such that $a(b,i) = \max(D) = h(D)$. We show that $C_{a(b,i)}
        = L$. 
        \begin{enumerate}
            \item[$\supseteq$:] To show $C_{a(b,i)} \supseteq L$, let $x \in L$.
                We show that $\varphi_{a(b,i)}(x) = 1$. As $x \in L$, there
                exists $k' \in \N$ and $b' \in \settwo{0,1}$ such that $x =
                a(b',k')$. If $b = b'$ and $i = k'$, then $\varphi_{a(b,i)}(x) =
                1$ by definition. Otherwise, as all $\sigma_j$ are defined, by
                Claim~\ref{Claim:fdecide}, we have $f(b',k') = 1$, which is
                exactly the output of $\varphi_{a(b,i)}(x)$.
            \item[$\subseteq$:] To show $C_{a(b,i)} \subseteq L$, let $x \notin
                L$. Now, either there exist no $k' \in \N$ and $b' \in
                \settwo{0,1}$ such that $x = a(b',k')$. Then,
                $\varphi_{a(b,i)}(x) = 0$ by definition. Else, let $k'$ and $b'$
                such that $x = a(b',k')$. If $b' = 1 - b$ and $k' = i$, then
                $\varphi_{a(b,i)}(x) = 0$ by definition. Otherwise, again as all
                $\sigma_j$ are defined, by Claim~\ref{Claim:fdecide}, $f(b', k')
                = 0$, which is exactly the output of $\varphi_{a(b,i)}(x)$. 
        \end{enumerate}
        Thus, $C_{a(b,i)} = L$. So, $h$ learns $L$. On the other hand, $h'$ does
        not, as it makes infinitely many mind changes on text $\tilde{T}$.
    \itemin{2. Case:} There exists $j$ such that $\sigma_j$ is defined, but
        $\sigma_{j+1}$ is not. Let $j'$ be minimal such. Let $m \coloneqq
        |\sigma_{j'}|$ and consider the texts
        \begin{align*}
          T_0 = {\sigma_{j'}} \concat a{(0,m)}^\infty, \\ 
          T_1 = {\sigma_{j'}} \concat a{(1,m)}^\infty,
        \end{align*}
        as well as the languages $L_0 = \content(T_0)$ and $L_1 = \content(T_1)$. We
        show that $h$ can learn both $L_0$ and $L_1$, while $h'$ cannot. To show
        that $h$ learns $L_0$, let $T \in \Txt(L_0)$. As $L_0$ is finite, there
        exists $n_0$ such that $\content(T[n_0]) = L_0$. Then, for all $n \geq n_0$,
        we have $h(\content(T[n])) = \max(\content(T[n])) = a(0,m)$ as $a$ is
        interleaved increasing. We show that $C_{a(0,m)} = L_0$. 
        \begin{enumerate}
          \item[$\supseteq$:] To show $C_{a(0,m)} \supseteq L_0$, let $x \in L_0$.
              If $x= a(0,m)$, then $x \in C_{a(0,m)}$ by definition of
              $\varphi_{a(0,m)}(x)$. Otherwise, there exist $k' < m$ and $b' \in
              \settwo{0,1}$ such that $x = a(b',k')$. Note that $i_0(k') < j'$. Thus
              we can apply Claim~\ref{Claim:fdecide} and get $f(b',k')=1$ which is
              exactly the output of $\varphi_{a(0,m)}(x)$.  
          \item[$\subseteq$:] To
              show $C_{a(0,m)} \subseteq L_0$, let $x \notin L_0$. Now, either there
              exist no $k' \in \N$ and $b' \in \settwo{0,1}$ such that $x =
              a(b',k')$. Then, $\varphi_{a(0,m)}(x) = 0$ by definition. Else, let
              $k'$ and $b'$ such that $x = a(b',k')$. We distinguish the following
              cases to show that $x \notin C_{a(0,m)}$.
              \begin{itemize}
                \item If $k' = m$, then $\varphi_{a(b,i)}(x) = 0$ by definition.
                \item In the case of $k' > m$, we have $k' > |\sigma_{i_0(k')}|$ and thus $f(b', k') = 0$.
                \item Given the case $k' < m$, again by Claim~\ref{Claim:fdecide}, $f(b', k') = 0$.
              \end{itemize}
        \end{enumerate}
        Thus, $C_{a(0,m)} = L_0$ as desired. The reasoning for $L_1$ is analogous.

        So, $h$ learns both $L_0$ and $L_1$. However, $h'$ converges to the same
        hypothesis on both $T_0$ and $T_1$ rendering it incapable to learn both
        languages simultaneously. \qedhere
  \end{enumerate}\noqed
\end{proof}

Next, we show that, just as for $W$-indices, a padding argument makes iterative
behaviorally correct learners as powerful as Gold-style ones. 

\begin{theorem}
  We have that ${[\Txt\It\Bc_C]}\rec = {[\Txt\G\Bc_C]}\rec$.
\end{theorem}

\begin{proof}
  The inclusion ${[\Txt\It\Bc_C]}\rec \subseteq {[\Txt\G\Bc_C]}\rec$ follows
  immediately. For the other, we apply a padding argument as in the proof of
  $[\Txt\It\Bc_W] = [\Txt\G\Bc_W]$ as given in \citet{KSS17}. Let
  $h\in\partialCp$ be a learner and let $\La = \Txt\G\Bc_C(h)\cap\REC$. Recall
  that
  $\pad \in \totalCp$ is a padding function, that is, for all $e \in \N$ and all
  finite sequences $\sigma$ we have $\varphi_e = \varphi_{\pad(e, \sigma)}$. For
  any finite sequence $\sigma$, we define the iterative learner ${(h')}^*(\sigma)
  = \pad(h(\sigma), \sigma)$. Intuitively, the learner $h'$ simulates $h$ in the
  following way. At every iteration, given a datum $x$ and its previous guess
  $\pad(h(\sigma), \sigma)$, the learner unpads $\sigma$, attaches $x$ to it and
  makes the guess $\pad(h(\sigma \concat x), \sigma \concat x)$. While this is
  syntactically different hypothesis, it has the same semantics as $h(\sigma
  \concat x)$.

  We show that $h'$ $\Txt\It\Bc_C$-learns $\La$. Let $L \in \La$ and $T \in
  \Txt(L)$. Then, for every $n\in\N$, we have $C_{h(T[n])} = C_{{(h')}^*(T[n])}$.
  Thus, $h'$ learns $L$ as $h$ does.
\end{proof}

We show that the classes of languages learnable by some behaviorally correct
Gold-style (or, equivalently, iterative) learner, can also be learned by
partially set-driven ones. We follow the proof of \citet{DoskocK20} after a
private communication with Sanjay Jain. The idea there is to search for minimal
$\Bc$-locking sequences without directly mimicking the $\G$-learner. We transfer
this idea to hold when converging to $C$-indices as well. We remark that, while
doing the necessary enumerations, one needs to make sure these are
characteristic. One obtains this as the original learner eventually outputs
characteristic indices. 

\begin{theorem}
  We have that ${[\Txt\Psd\Bc_C]}\rec = {[\Txt\G\Bc_C]}\rec$. \label{thm:Cind-Psd-G-Bc}\label{th:psdgbc} 
\end{theorem}

\begin{proof}
  The inclusion ${[\Txt\Psd\Bc_C]}\rec \subseteq {[\Txt\G\Bc_C]}\rec$ follows
  immediately. For the other, we follow an idea how $\Txt\G\Bc$-learning can be
  made partially set-driven, as given in \citet{DoskocK20} following a private
  communication with Sanjay Jain. To that end, let $h$ be a learner and let $\La
  = \Txt\G\Bc_C(h)\cap\REC$. By \citet{KP16}, we may assume $h$ to be total.
  Now, define the $\Psd$-learner $h'$ as follows. For $x, a \in \N$ and for
  finite $D \subseteq \mathbb{N}$ and $t \geq 0$, we first define the auxiliary
  total predicate $Q(x, a, (D,t))$ which holds true if and only if there exists
  a sequence $\sigma \in D^{\leq t}_\#$ such that both
  \begin{enumerate}[label=\textnormal{(\arabic*)}]
    \item for all $\tau \in D^{\leq t}_\#$ we have that $\varphi_{h(\sigma\tau)}(x) = a$, and\label{ForwC}
    \item for all $\sigma' < \sigma$, with $\sigma' \in D^*_\#$, there exists $\tau' \in D^{\leq t}_\#$ such that $\varphi_{h(\sigma'\tau')}(x) = a$.\label{BackwC}
  \end{enumerate}
  With the help of $Q$ we define the learner $h'$ such that, for finite $D \subseteq \mathbb{N}$, $t \geq 0$ and for all $x \in \N$,
  \begin{align*}
    \varphi_{h'(D,t)}(x) = \begin{cases} 
      1, &\falls Q(x, 1,(D,t)); \\ 0, &\sonst.
    \end{cases}
  \end{align*}
  Intuitively, we check whether the information given is enough to witness a
  (minimal) $\Bc_C$-locking sequence. Then, for every element, we evaluate
  whether it belongs to the language or not. Note that upon correct learning, no
  element can be witnessed to be both part of the language and not part of it. 

  We first show that $h'(D,t)$ is well defined. Assume there exists some $x\in\N$ and some natural number
  $a \neq 1$ such that $Q(x,a,(D,t))$ and $Q(x,1,(D,t))$ simultaneously,
  witnessed by $\sigma_a$ and $\sigma_1$ respectively. Without loss of
  generality, suppose $\sigma_a < \sigma_1$. Then, by Condition~\ref{BackwC} of
  $\sigma_1$, there exists some $\tau' \in D^{\leq t}_\#$ such that
  $\varphi_{h(\sigma_a \tau')}(x) = 1$. However, by Condition~\ref{ForwC} of
  $\sigma_a$, for all $\tau \in D^{\leq t}_\#$, we have $\varphi_{h(\sigma_a
  \tau)}(x) = a$, a contradiction.

  Let $L \in \La$. We proceed by proving $L \in \Txt\Psd\Bc_C(h')$. For that,
  let $T \in \Txt(L)$. By \citet{BlumBlum75}, there exists a $\Bc_C$-locking
  sequence for $h$ on $L$. Let $\alpha$ be the least such $\Bc_C$-locking
  sequence with respect to $<$. By \citet{OSW86}, for each $\alpha' < \alpha$
  such that $\content(\alpha') \subseteq L$, there exists $\tau_{\alpha'}$ such
  that $\alpha'\tau_{\alpha'}$ is a $\Bc_C$-locking sequence for $h$ on $L$.
  Now, let $n_0\in\N$ be large enough such that
  \begin{itemize}
    \item $n_0 \geq |\alpha|$,
    \item $\content(\alpha) \subseteq \content(T[n_0])$ and
    \item for all $\alpha' < \alpha$ such that $\content(\alpha') \subseteq L$,
        we have $\content(\alpha'\tau_{\alpha'}) \subseteq \content(T[n_0])$ and
        $|\tau_{\alpha'}| \leq n_0$.
  \end{itemize}
  We claim that for $t \geq n_0$ and $D=\content(T[t])$, we have $C_{h'(D,t)} =
  L$. Let $x \in \N$ and $a\in \N$ such that $\chi_L(x) = a$. As
  $D$ and $t$ are chosen sufficiently large, $\alpha$ is a candidate for the
  enumeration of $C_{h'(D,t)}$. Since $\alpha$ is a $\Bc_C$-locking sequence,
  for every $\tau \in D_\#^{\leq t}$, we will witness $\varphi_{h(\alpha
  \tau)}(x) = a$. Thus, Condition~\ref{ForwC} is witnessed. On the other hand,
  observe that for every $\sigma' < \alpha$, with $\content(\sigma') \subseteq
  D$, we have $\tau_{\sigma'} \in D_\#^{\leq t}$. So, we will witness
  $\varphi_{h(\sigma' \tau_{\sigma'})}(x) = a$ for some $\tau_{\sigma'} \in
  D_\#^{\leq t}$, that is, the Condition~\ref{BackwC}.

  As an element cannot be witnessed to be part of the language and not part of
  it simultaneously, we finally have \(\chi_{L}=\varphi_{h'(D,t)}\), concluding
  the proof, as $L \in \Txt\Psd\Bc_C(h')$.
\end{proof}

Lastly, we investigate transductive learners. Such learners base their
hypotheses on a single element. Thus, one would expect them to benefit from
dropping the requirement to converge to a \emph{single} hypothesis.
Interestingly, this does not hold true. This surprising fact originates from
$C$-indices encoding characteristic functions. Thus, one can simply search for
the minimal element on which no ``?'' is conjectured. The next result finalizes
the map shown in Figure~\ref{Map:Convergence} and, thus, this section.

\begin{theorem}
  We have that ${[\Txt\Td\Ex_C]}\rec = {[\Txt\Td\Bc_C]}\rec$.
\end{theorem}

\begin{proof}
  The inclusion ${[\Txt\Td\Ex_C]}\rec \subseteq {[\Txt\Td\Bc_C]}\rec$ is
  immediate. For the other direction, let $h$ be a learner and $\La =
  \Txt\Td\Bc_C(h)\cap\REC$. We provide a learner $h'$ such that $\La \subseteq
  \Txt\Td\Ex_C(h')$. Let $L \in \La$. We note that, for any $x \in L$, if not
  $h(x) = \mbox{?}$, then $h(x)$ is a $C$-index for the language $L$, that is,
  $C_{h(x)} = L$. Assume there exists an $x \in L$ where $h(x) \neq \mbox{?}$ is
  no $C$-index of $L$. Then, on a text with infinitely many occurrences of $x$
  the language $L$ cannot be $\Txt\Td\Bc_C$-learned using $h$. Now, we define
  the $\Txt\Td\Ex_C$-learner $h'$ for all $x \in \N$ as
  \[
      h'(x) = \begin{cases} \mbox{?}, &\falls h(x) = \mbox{?}; \\ h(\min\{x \in
      C_{h(x)} \mid h(x) \neq \mbox{?} \}), &\sonst. \end{cases}
  \]
  It is straightforward to verify the correctness of $h'$. 
\end{proof}

\acks{%
  This work was supported by DFG Grant Number KO 4635/1-1. 
}

\bibliography{LTbib}

\end{document}